\documentclass[twocolumn,aps,pra,superscriptaddress,showpacs,floatfix]{revtex4-1}

\usepackage[utf8x]{inputenc}
\usepackage[english]{babel}
\usepackage[T1]{fontenc}
\usepackage{amssymb}
\usepackage{bbold}
\usepackage{amsmath}
\usepackage{amsthm}
\usepackage[pdftex]{graphicx}
\usepackage{epstopdf}
\usepackage{pgfplots}
\usepackage[colorlinks=true,linkcolor=black]{hyperref}

\newtheorem*{mydef}{Definition}
\newtheorem{mythe}{Theorem}

\begin{document}

\title{Optimal state pairs for non-Markovian quantum dynamics}

\author{Steffen Wi{\ss}mann}
\email{steffen.wissmann@jupiter.uni-freiburg.de}
\affiliation{Physikalisches Institut, Universit\"at Freiburg, 
Hermann-Herder-Stra{\ss}e 3, D-79104 Freiburg, Germany}

\author{Antti Karlsson}
\email{aspkar@utu.fi}
\affiliation{Turku Centre for Quantum Physics, Department of Physics and 
Astronomy, University of Turku, FI-20014 Turun yliopisto, Finland}

\author{Elsi-Mari Laine}
\affiliation{Turku Centre for Quantum Physics, Department of Physics and 
Astronomy, University of Turku, FI-20014 Turun yliopisto, Finland}

\author{Jyrki Piilo}
\affiliation{Turku Centre for Quantum Physics, Department of Physics and 
Astronomy, University of Turku, FI-20014 Turun yliopisto, Finland}

\author{Heinz-Peter Breuer}
\affiliation{Physikalisches Institut, Universit\"at Freiburg, 
Hermann-Herder-Stra{\ss}e 3, D-79104 Freiburg, Germany}

\date{\today}

\begin{abstract}
We study a recently proposed measure for the quantification of quantum 
non-Markovianity in the dynamics of open systems which is based on the 
exchange of information between the open system and its environment. This 
measure relates the degree of memory effects to certain optimal initial state pairs 
featuring a maximal flow of information from the environment back to the
open system. We rigorously prove that the states of these optimal pairs must lie on 
the boundary of the space of physical states and that they must be
orthogonal. This implies that quantum memory effects are maximal for states 
which are initially distinguishable with certainty, having a maximal information 
content. Moreover, we construct an explicit example which demonstrates that 
optimal quantum states need not be pure states.
\end{abstract}


\maketitle

\section{Introduction}\label{sec:intro}

Many quantum physical systems are modelled approximately as closed systems. 
However, in most practical cases the interaction of a given system with its 
environment can not be neglected and this makes it necessary to use the theory of 
open quantum systems~\cite{BREUERBOOK}.
In the past, it has been very common to neglect quantum memory effects and 
to resort to a Markovian approximation. In many applications the time evolution of 
the open system is described by a quantum dynamical semigroup represented by
a generator of the Lindblad-Gorini-Kossakowski-Sudarshan 
form~\cite{Gorini,Lindblad}.

However, when the environment has a nontrivial structure, one has to account for 
memory effects influencing the dynamics of the open system. During recent 
years, there has been significant conceptual, theoretical, and experimental 
progress dealing with non-Markovian processes~\cite
{Piilo08,Vacchini08,Wolf,Lidar2009,Measure,MeasurePaper,Fisher,RHP, 
Kossakowski2010,Paternostro,NJP,Maniscalco,Liu,Tang,Haikka,HRP,SpIssue}. 
In particular, the very definition of non-Markovianity 
and quantification of quantum memory effects in the dynamics of open systems 
has received a lot of interest. Several measures for quantum non-Markovianity 
have been proposed which are based on different mathematical and physical 
concepts~\cite{Wolf,Measure,MeasurePaper,Fisher,RHP}.

In the present paper we investigate the non-Markovianity measure proposed in
\cite{Measure} which expresses the degree of memory effects in terms of the 
amount of information exchanged between the open system and its environment. 
Within this approach the degree of non-Markovianity is connected to certain 
optimal initial pairs of quantum states, leading to a maximal flow of information 
from the environment to the open system. Here, we will study the mathematical 
and physical properties of such optimal state pairs.

Throughout this paper $\mathcal{H}$ refers to the Hilbert space of the open 
quantum system and the corresponding set of physical states, i.e.~the set of 
positive trace class operators $\rho$ with unit trace is denoted by 
$\mathcal{S}(\mathcal{H})$. We assume that the dynamics of the open quantum 
system can be represented in terms of a one-parameter family 
\begin{equation} \label{family}
 \Phi=\{\Phi_t \mid 0 \leq t \leq T\}
\end{equation} 
of completely positive and trace preserving (CPT) linear maps $\Phi_t$, where 
$\Phi_0$ is equal to the identity map \cite{BREUERBOOK}. We will use the 
quantity
\begin{equation}\label{eq:measure}
 \mathcal{N}(\Phi)\equiv \max_{\rho_{1,2}\in \mathcal{S}(\mathcal{H})} 
 \int_{\sigma>0} dt ~\sigma(t,\rho_1,\rho_2)
\end{equation}
as a measure for the degree of memory effects in the open system dynamics
\cite{Measure,MeasurePaper}. Here, 
\begin{equation}\label{eq:def-sigma}
 \sigma(t,\rho_1,\rho_2) \equiv \frac{d}{dt}
 \mathcal{D}\bigl(\Phi_t(\rho_1),\Phi_t(\rho_2)\bigr),
\end{equation}
where $\mathcal{D}$ denotes the trace distance \cite{Hayashi,Nielsen} defined by 
\begin{equation}
 \mathcal{D}(\rho_1,\rho_2) = \frac{1}{2} {\mathrm{Tr}}|\rho_1-\rho_2|.
\end{equation} 
By definition, the non-Markovianity measure $\mathcal{N}(\Phi)$ is a positive 
functional of the process $\Phi$. It is zero if and only if the trace distance between 
any pair of initial states is a monotonically decreasing function of time $t$ 
signifying a continuous loss of information from the open system to its 
environment. A nonzero value for the measure means that there is an initial state 
pair for which the trace distance increases over a certain time interval which can 
be interpreted as a flow of information from the environment back to the open 
system implying the presence of memory effects. The time integral in 
\eqref{eq:measure} determines the total backflow of information for a certain pair 
of initial states $\rho_1$ and $\rho_2$. The quantity $\mathcal{N}(\Phi)$ is then 
found by taking the maximum over all initial state pairs, i.e. $\mathcal{N}(\Phi)$ 
represents the maximal possible backflow of information. 

A pair of states $\rho_1,\rho_2 \in \mathcal{S}(\mathcal{H})$ is said to be an 
{\textit{optimal state pair}} if the maximum in Eq.~\eqref{eq:measure} is attained 
for this pair. In the present paper we are interested in the mathematical and 
physical features of the maximization procedure involved in expression 
\eqref{eq:measure}. In particular, we will obtain general statements about the 
physical properties of optimal state pairs which give rise to the maximal possible 
degree of memory effects in their dynamics. To this end, we first show in 
Sec.~\ref{sec:boundary} that optimal pairs of states must lie on the boundary of
the state space $\mathcal{S}(\mathcal{H})$. We then proceed in
Sec.~\ref{Orthogonality} to demonstrate that the states of any optimal pair must be
orthogonal. This is our central result which is physically very plausible since
it implies that an optimal pair must have the largest possible initial trace
distance $\mathcal{D}(\rho_1,\rho_2)=1$. Thus, the maximal flow of information 
from the environment back to the open system emerges if the initial state pair is 
distinguishable with certainty, i.e.~has a maximal information content. 
We discuss specific features of the case of infinite-dimensional state spaces in 
Sec.~\ref{Infinite}, and present in Sec.~\ref{translatable} an alternative proof for 
the orthogonality of optimal state pairs which employs the joint translatability of
nonorthogonal states.

In the simple case of a qubit, the orthogonality of the optimal pair implies that 
both states of the pair must be pure. In Sec.~\ref{sec:counterex1} we
study the question of whether this statement holds true for higher dimensional
systems. The answer is negative. In fact, we will construct in this
section an example for the dynamics of a three-level system ($\Lambda$-system)
for which the optimal pair is not a pure state pair. Thus, optimal state pairs
exhibiting maximal backflow of information can indeed be mixed states. Finally, we
draw some conclusions from our results in Sec.~\ref{sec:conclusion}.

\section{Characterizing optimal pairs of quantum states}\label{sec:optimalpair}

\subsection{Restriction to the boundary of the state space}\label{sec:boundary}

The state space $\mathcal{S}(\mathcal{H})$ of an open quantum system is
given by the set of positive trace class operators on its Hilbert space 
$\mathcal{H}$ with trace one \cite{Heinosaari},
\begin{equation}
 \mathcal{S}(\mathcal{H}) = \left\{ \rho \mid \rho \geq 0, \;
 \mathrm{Tr} \rho = 1 \right\}.
\end{equation} 
On the basis of the convex structure of the state space one can define the 
boundary $\partial\mathcal{S}(\mathcal{H})$ of $\mathcal{S}(\mathcal{H})$
as follows. A point $\rho \in \mathcal{S}(\mathcal{H})$ is defined to be an interior point if and only if for all $\sigma \in \mathcal{S}(\mathcal{H})$ there
is a real number $\lambda > 1$ such that
\begin{equation} \label{eq:defbound}
 (1-\lambda)\sigma + \lambda\rho \in \mathcal{S}(\mathcal{H}). 
\end{equation} 
Denoting the set of all interior points by $\tilde{\mathcal{S}}(\mathcal{H})$, 
one defines the boundary of the state space by
\begin{equation}
 \partial\mathcal{S}(\mathcal{H}) = 
 \mathcal{S}(\mathcal{H}) \setminus \tilde{\mathcal{S}}(\mathcal{H}). 
\end{equation} 
Thus, $\rho \in \partial\mathcal{S}(\mathcal{H})$ if and only if there exists 
$\sigma \in \mathcal{S}(\mathcal{H})$ such that for all $\lambda > 1$ the
operator $(1-\lambda)\sigma + \lambda\rho$ does
not belong to the state space $\mathcal{S}(\mathcal{H})$. 

It is easy to show that a density matrix $\rho \in \mathcal{S}(\mathcal{H})$ belongs
to the boundary $\partial\mathcal{S}(\mathcal{H})$ of the state space if it has
a zero eigenvalue. In fact, let $|\varphi\rangle$ be a normalized eigenvector 
of $\rho$ with eigenvalue zero, and $P=|\varphi\rangle\langle\varphi|$ the
corresponding projection operator. Then, for all $\lambda > 1$ the operator
$(1-\lambda)P + \lambda\rho$ has the negative eigenvalue $1-\lambda < 0$
with corresponding eigenvector $|\varphi\rangle$, and thus does not belong
to the state space.

For finite-dimensional systems the boundary $\partial\mathcal{S}(\mathcal{H})$ 
actually contains only states with zero eigenvalues. This can be seen by setting 
$\lambda=1+\epsilon$ with $\epsilon > 0$ in \eqref{eq:defbound} which yields 
the operator $\rho'=(1+\epsilon)(\rho-\epsilon(1+\epsilon)^{-1}\sigma)$. If all 
eigenvalues of $\rho$ are strictly positive, it follows from the continuity of the roots 
of the characteristic polynomial that for sufficiently small $\epsilon$ also all
eigenvalues of $\rho'$ are positive. This shows that states 
with strictly positive eigenvalues are in the interior $\tilde{\mathcal{S}}(\mathcal{H})
$ such that $\partial\mathcal{S}(\mathcal{H})=\{\rho\in\mathcal{S}(\mathcal{H})|0\in 
\mathrm{spec}(\rho)\}$ for finite-dimensional systems.

As has been mentioned already in the introduction a pair of states 
$\rho_1,\rho_2$ is said to be an optimal state pair if the maximum in 
Eq.~\eqref{eq:measure} is attained for this pair, i.e., if
\begin{equation}\label{eq:optimality}
 \mathcal{N}(\Phi) =  \int_{\sigma>0} dt ~\sigma(t,\rho_1,\rho_2)
\end{equation}
holds for this pair. If we speak of an optimal pair we will always assume that 
the quantum process under consideration is non-Markovian, i.e. 
$\mathcal{N}(\Phi)>0$, since otherwise any pair of quantum states would be
trivially optimal. Moreover, we note that the optimality property implies of course 
that $\rho_1\neq\rho_2$.

\begin{mythe}\label{the:boundary}
Let $\rho_1,\rho_2\in\mathcal{S}(\mathcal{H})$ be an optimal state pair. Then
both states lie on the boundary of the states space, 
$\rho_1,\rho_2\in\partial\mathcal{S}(\mathcal{H})$.
\end{mythe}

\begin{proof}
Suppose that (at least) one state of the pair,
say $\rho_2$, does not belong to the boundary. Hence, $\rho_2$ is
an interior point and there exists $\lambda > 1$ such that 
\begin{equation}\label{rho3}
 \rho_3 = (1-\lambda)\rho_1 + \lambda\rho_2 \in \mathcal{S}(\mathcal{H})
\end{equation}
is a quantum state (see Fig.~\ref{fig:1}). The time evolution of the three states is 
given by $\rho_i(t) = \Phi_t(\rho_i)$, $i=1,2,3$. By the linearity of the dynamical 
map we have
\begin{equation}\label{rho3t}
 \rho_3(t) = (1-\lambda)\rho_1(t) + \lambda\rho_2(t)
\end{equation}
and, hence,
\begin{equation}
 \rho_1(t) - \rho_3(t) = \lambda (\rho_1(t)-\rho_2(t)).
\end{equation}
It follows that
\begin{equation}
 \mathcal{D}(\rho_1(t),\rho_3(t)) = \lambda \mathcal{D}(\rho_1(t),\rho_2(t)). 
\end{equation}
Note that $\lambda$ is a fixed number strictly larger than $1$. Thus, the last
equation tells us that the trace distance between $\rho_1(t)$ and $\rho_3(t)$
is always larger by the constant factor $\lambda$ than the trace distance between 
$\rho_1(t)$ and $\rho_2(t)$. This implies that the quantity
$\int_{\sigma>0} dt ~\sigma(t,\rho_1,\rho_3)$ is larger than
$\int_{\sigma>0} dt ~\sigma(t,\rho_1,\rho_2)$ by the same factor $\lambda$
(see Eq.~\eqref{eq:def-sigma}). It follows that $\rho_1,\rho_2$ cannot be an 
optimal pair which is a contradiction. Consequently, any optimal pair of states must 
belong to the boundary of the state space.
\end{proof}

\begin{figure}[tbh]
 \centering
 \includegraphics[width=0.25\textwidth]{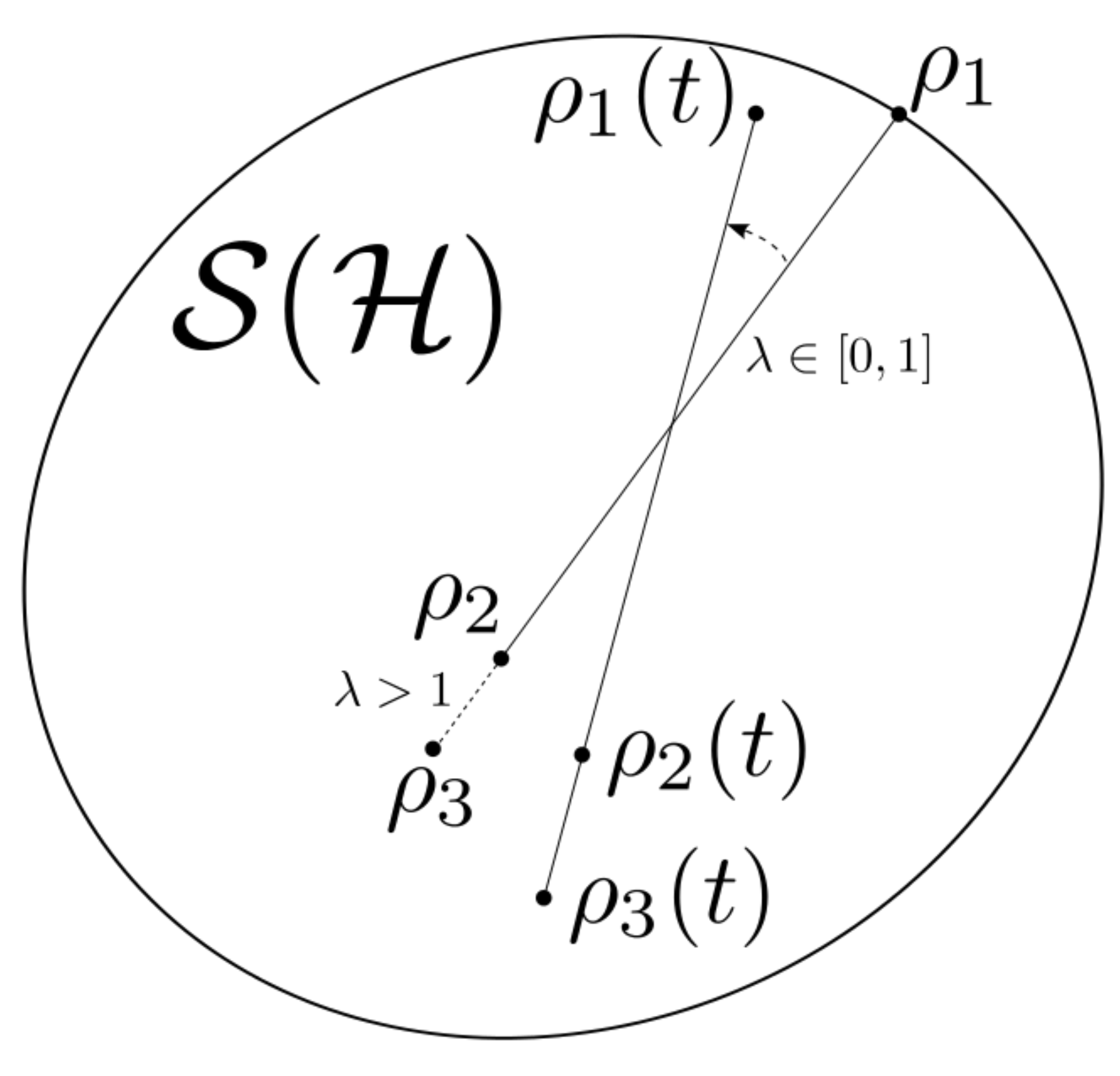}
 \caption{Illustration of the decomposition \eqref{rho3} and of the corresponding
 time evolution given by Eq.~\eqref{rho3t}.}\label{fig:1}
\end{figure}

Thus we see that the maximization over all initial state pairs in the definition 
\eqref{eq:measure} of the non-Markovianity measure $\mathcal{N}(\Phi)$ can be
restricted to the boundary $\partial\mathcal{S}(\mathcal{H})$ of the state space,
\begin{equation}
 \mathcal{N}(\Phi) = \max_{\rho_{1,2}\in \partial\mathcal{S}({\mathcal{H}}_S)} 
 \int_{\sigma>0} dt ~\sigma(t,\rho_1,\rho_2).
\end{equation}
We emphasize that this result is valid for any Hilbert space $\mathcal{H}$ and 
for any family of linear dynamical maps $\Phi_t$. 
The linearity guarantees the invariance of the decomposition defined in 
\eqref{rho3} for all times $t$ (see Fig.~\ref{fig:1}).

\subsection{Orthogonality of optimal state pairs}\label{Orthogonality}

In this section we demonstrate that optimal state pairs must be 
orthogonal which strengthens the result of Sec.~\ref{sec:boundary}. We recall that
two quantum states $\rho_1$ and $\rho_2$ are said to be orthogonal, 
$\rho_1\perp\rho_2$, if and only if their supports are orthogonal, 
$\mathrm{supp}(\rho_1)\perp\mathrm{supp}(\rho_2)$, where the support is
defined as the subspace spanned by the eigenvectors with nonzero eigenvalues.
Note that the orthogonality of two states implies that both states have a zero
eigenvalue and, hence, that both states belong to the boundary of the state space.
We also mention that the trace distance between any two quantum states
satisfies $\mathcal{D}(\rho_1,\rho_2) \leq 1$, where the equality sign holds
if and only if $\rho_1\perp\rho_2$. Orthogonality is thus equivalent to unit
trace distance \cite{Hayashi,Nielsen}.

\begin{mythe}\label{the:orthogonality}
Optimal state pairs are orthogonal.
\end{mythe}

\begin{proof}
Let $\rho_1,\rho_2\in\mathcal{S}(\mathcal{H})$ be an optimal pair of states, and
suppose that the states are not orthogonal, $\rho_1\not\perp\rho_2$. According
to the Jordan-Hahn decomposition \cite{Nielsen} there exist positive and 
orthogonal operators 
$P_1$ and $P_2$ such that
\begin{equation}\label{J-H-Decomp}
 \rho_1 - \rho_2 = P_1 - P_2.
\end{equation}
Taking the trace of this equation we see that 
$\lambda \equiv \mathrm{Tr} P_1 = \mathrm{Tr} P_2$. From $\rho_1\neq\rho_2$
we conclude that $\lambda > 0$. Using Eq.~\eqref{J-H-Decomp}, the 
orthogonality of $P_1$ and $P_2$, and the fact that the trace distance
between nonorthogonal states is always strictly smaller than $1$, we find
\begin{equation}
 1 > \mathcal{D}(\rho_1,\rho_2) = \mathcal{D}(P_1,P_2)
 = \frac{1}{2} \left( \mathrm{Tr} P_1 + \mathrm{Tr} P_2 \right) = \lambda.
\end{equation}
Thus, we have $0 < \lambda < 1$. Now we define the operators
$\sigma_1 = P_1/\lambda$ and $\sigma_2 = P_2/\lambda$.
Being positive and of unit trace, these operators represent quantum states.
Moreover we have
\begin{equation}
 \sigma_1 - \sigma_2 = \frac{1}{\lambda} (\rho_1 - \rho_2).
\end{equation}
By use of the linearity of the dynamical maps we obtain
\begin{equation}
 \sigma_1(t) - \sigma_2(t) = \frac{1}{\lambda} (\rho_1(t) - \rho_2(t)),
\end{equation}
from which it follows that
\begin{equation}
 \mathcal{D}(\sigma_1(t),\sigma_2(t)) 
 = \frac{1}{\lambda} \mathcal{D}(\rho_1(t),\rho_2(t)).
\end{equation}
Since $\lambda^{-1}>1$ we can conclude from the last equation that the pair
$\sigma_1,\sigma_2$ yields a non-Markovianity which is strictly larger than
that of the pair $\rho_1,\rho_2$, which contradicts the assumption that
$\rho_1,\rho_2$ is an optimal pair. Hence, $\rho_1$ and $\rho_2$ must be
orthogonal.
\end{proof}

It follows from theorem \ref{the:orthogonality} that the maximization in the 
definition \eqref{eq:measure} of the measure $\mathcal{N}(\Phi)$ for quantum 
non-Markovianity can be restricted to orthogonal initial state pairs. Again, this 
result holds for any Hilbert space $\mathcal{H}$ and any family of linear dynamical 
maps $\Phi_t$.

\subsection{Infinite dimensional state spaces}\label{Infinite}

It is important to emphasize that for infinite dimensional spaces an optimal state 
pair need not necessarily exist. To be mathematically more precise, the maximum 
in the definition \eqref{eq:measure} for the non-Markovianity measure should be 
replaced by the supremum to account for such cases. 

In theorems \ref{the:boundary} and \ref{the:orthogonality} we have assumed that 
an optimal pair of states does exist. If there is no such pair
there is a sequence of state pairs $\rho^n_1,\rho^n_2$ such that 
$\int_{\sigma>0} dt ~\sigma(t,\rho^n_1,\rho^n_2)$ converges to 
$\mathcal{N}(\Phi)$ for $n \rightarrow \infty$. Employing the construction used 
in the proof of theorem \ref{the:orthogonality} one can show that the pairs
$\rho^n_1,\rho^n_2$ can always be taken to be orthogonal. This means that
the non-Markovianity measure can be approximated with arbitrary precision by 
orthogonal state pairs and that we can write
\begin{equation}
 \mathcal{N}(\Phi) = \sup_{\rho_1\perp\rho_2} 
 \int_{\sigma>0} dt ~\sigma(t,\rho_1,\rho_2).
\end{equation}

\subsection{Orthogonality and parallel translations}\label{translatable}

In this section we present an alternative proof for the orthogonality of optimal
state pairs, where we restrict ourselves to finite dimensional Hilbert spaces,
$\mathrm{dim}\,\mathcal{H}=N$. This proof could be of interest also in other 
contexts since it relies on the behavior of pairs of states under parallel 
translations.

The idea of the proof is based on the observation gained from two-level systems 
that nonorthogonal states on the boundary of the state space can be 
simultaneously translated by a traceless Hermitian operator to yield a pair of mixed 
states, while the trace distance of the pair is invariant under such translations. To 
make this idea more precise we define 
\begin{equation}
 \mathcal{E}(\mathcal{H}) =
 \{ A \mid A\neq0, \, A=A^\dag, \, \mathrm{Tr}A=0 \}
\end{equation}
to be the set of nonzero, Hermitian and traceless operators on $\mathcal{H}$.

\begin{mydef}\label{def1:trans}
Two states $\rho_1,\rho_2\in\mathcal{S}(\mathcal{H})$ are called jointly 
translatable if and only if there exists an $A\in\mathcal{E}(\mathcal{H})$
such that $\rho_k-A\in\mathcal{S}(\mathcal{H})$ for $k=1,2$.
\end{mydef}

Hence, two states are said to be jointly translatable if and only if there is
a nontrivial Hermitian and traceless operator $A$ which can be subtracted from
the states without leaving the state space. We prove that any pair of 
nonorthogonal states is jointly translatable in such a way that both translated 
states do not belong to the boundary of the state space.

\begin{mythe}\label{the:trans}
If $\rho_1,\rho_2\in\mathcal{S}(\mathcal{H})$ and $\rho_1\not\perp\rho_2$ then  
$\rho_1,\rho_2$ are jointly translatable. Moreover, there exists an operator
$A\in\mathcal{E}(\mathcal{H})$ such that 
$\rho_k-A\notin\partial\mathcal{S}(\mathcal{H})$ for $k=1,2$.
\end{mythe}

\begin{proof}
Let $\rho_{1,2}$ be given in terms of their spectral decomposition,
\begin{equation}
 \rho_k=\sum_{i=1}^N p_i^{(k)}|\psi_i^{(k)}\rangle\langle\psi_i^{(k)}|,
\end{equation}
where $p_i^{(k)}$, $ i=1,\dots,N$, denotes the eigenvalues and 
$|\psi_i^{(k)}\rangle$ the corresponding eigenvectors of $\rho_k$. The
assumption $\rho_1\not\perp\rho_2$ implies that $\rho_1$ and $\rho_2$ have at 
least one eigenvector with nonzero eigenvalue which are not orthogonal. After 
possible relabeling one can assume that 
$\alpha\equiv\langle\psi_1^{(1)}|\psi_1^{(2)}\rangle\neq0$ and 
$p_1^{(1)},p_1^{(2)}\neq0$ holds. Moreover, by an appropriate choice of the 
phases of the eigenstates one can assume without restriction that $\alpha$ is real 
and positive, i.e.~$0<\alpha\leq1$. We consider the following superpositions of the 
two overlapping eigenvectors,
\begin{equation}
 |\psi_{\pm}\rangle=c_\pm\left(|\psi_1^{(1)}\rangle\pm|\psi_1^{(2)}\rangle\right),
\end{equation}
where the normalization constants obey $c_\pm^{-1}=\sqrt{2(1\pm \alpha)}$. The
projections onto these states are denoted by $P_\pm$. By positivity of the 
overlap $\alpha$ one has $c_+<c_-$. Now, we define
\begin{align}
 A_\epsilon&\equiv~\epsilon \cdot B~,\nonumber\\
B&=P_{+}-\Bigl(\frac{c_+}{c_-}\Bigr)^{2} P_{-}-\Bigl[1-\Bigl(\frac{c_+}{c_-}\Bigr)^
{2}\Bigr]\cdot \frac{1}{N}\mathbb{1}_N, \label{eq:shiftop}
\end{align}
where $\epsilon > 0$ is a real number to be chosen later.
Clearly, $B^\dag=B$, $B\neq0$ and $\mathrm{Tr}B =0$ so that $A_\epsilon$ is a 
candidate for the traceless hermitian operator which simultaneously shifts both 
states. Furthermore, as $c_+<c_-$ one recognizes that the two last terms are 
negative semidefinite while the first one is positive semidefinite. In order to show 
positivity of $\hat\rho_k\equiv\rho_k-A_\epsilon$ for an appropriate choice of 
$\epsilon$ we consider the quantity $\langle\chi|\hat\rho_k|\chi\rangle$
for an arbitrary normalized vector $|\chi\rangle\in\mathcal{H}$. We find
\begin{align}\label{eq:ortho1}
 \langle\chi|\hat\rho_k|\chi\rangle=& B_k(|\chi\rangle)+p_1^{(k)}\bigl|\langle\chi|
 \psi_1^{(k)}\rangle\bigr|^2-\epsilon\langle\chi|P_+|\chi\rangle\nonumber\\
 &+\epsilon\Bigl(\frac{c_+}{c_-}\Bigr)^{2}\langle\chi|P_{-}|\chi\rangle+\frac{\epsilon}
 {N}\Bigl[1-\Bigl(\frac{c_+}{c_-}\Bigr)^{2}\Bigr]\nonumber\\
 =&~B_k(|\chi\rangle)~+p_1^{(k)}\bigl|\langle\chi|\psi_1^{(k)}\rangle\bigr|^2
 \nonumber\\
 &+4c_+^2\epsilon\Bigl\{\frac{\alpha}{N}-\mathrm{Re}\bigl(\langle\psi_1^{(1)}|\chi 
 \rangle~ \langle\chi|\psi_1^{(2)}\rangle\bigr)\Bigr\}\nonumber\\
 \geq&~B_k(|\chi\rangle)~+p_1^{(k)}\bigl|\langle\chi|\psi_1^{(k)}\rangle\bigr|^2
 \nonumber\\
 &+4c_+^2\epsilon\Bigl\{\frac{\alpha}{N}-|\langle\chi|\psi_1^{(k)}\rangle|\Bigr\}~,
\end{align}
where $~B_k(|\chi\rangle)\equiv\sum_{i=2}^N p_i^{(k)}\bigl|\langle\chi|\psi_i^{(k)}
\rangle\bigr|^2\geq0$ and since $c_+^{-2}-c_-^{-2}=4\alpha$ by definition of $c_
\pm$. In the last step one uses the fact that $\mathrm{Re}(z)\leq|z|$ for all 
complex numbers $z$, and that $|\langle\psi_1^{(m)}|\chi\rangle|\leq 1$ for all 
normalized vectors by the Cauchy-Schwarz inequality. Consider now the 
quadratic function
\begin{equation}\label{eq:ortho4}
 g_\epsilon(x)=p_1^{(k)}x^2+4c_+^2\epsilon\Bigl\{\frac{\alpha}{N}-x\Bigr\}.
\end{equation}
One can show that this function is strictly positive for $k=1,2$ if $\epsilon$ satisfies
\begin{equation}\label{eq:ortho7}
 0 < \epsilon < \frac{\alpha}{Nc_+^2}\min_{k=1,2}p_1^{(k)}.
\end{equation}
Thus, if we choose any $\epsilon$ satisfying \eqref{eq:ortho7} we obtain 
$g_\epsilon(|\langle\chi|\psi_1^{(k)}\rangle|)>0$ for all normalized vectors 
$|\chi\rangle\in\mathcal{H}$ and $k=1,2$ which leads to
$\langle\chi|\hat\rho_{1,2}|\chi\rangle>0$. This demonstrates that $\hat\rho_{1,2}$ 
is positive and has no zero eigenvalue, i.e. 
$\hat\rho_{1,2}\in\tilde{\mathcal{S}}(\mathcal{H})$ due to the characterization of 
the boundary in terms of the eigenvalues given in Sec. \ref{sec:boundary}.  In 
particular, $\rho_1$ and $\rho_2$ are jointly translatable. 
\end{proof}

We remark that the converse of the statement of theorem \ref{the:trans} holds, too. 
Thus, two states are jointly translatable if and only if they are not orthogonal.

Theorem \ref{the:orthogonality} can now be proven as follows.
Let $\rho_1,\rho_2$ be an optimal pair, and 
suppose that $\rho_1\not\perp\rho_2$. Then by theorem \ref{the:trans} there exists
an operator $A\in\mathcal{E}(\mathcal{H})$ such that 
$\hat\rho_k\equiv\rho_k-A\notin\partial\mathcal{S}(\mathcal{H_S})$. Applying 
theorem \ref{the:boundary} we conclude that the states $\hat\rho_k$ are not
optimal. Since $\rho_1-\rho_2=\hat\rho_1-\hat\rho_2$ it follows that also the
states $\rho_k$ are not optimal, which represents a contradiction. Therefore, the 
optimal pair has to be orthogonal.

\section{Purity of optimal pairs}\label{sec:counterex1}

As a simple application of theorem \ref{the:orthogonality} one obtains the result 
that for all non-Markovian quantum processes of a two-dimensional system (qubit) 
the maximal backflow of information occurs for a pair of pure, orthogonal initial 
states, corresponding to antipodal points on the surface of the Bloch sphere. This 
follows immediately from the fact that for qubits the set of pure states is identical to 
the boundary $\partial\mathcal{S}(\mathcal{H})$ of the state space. For 
higher-dimensional systems this is no longer true, i.e., the set of pure states 
represents a proper subset of the boundary in this case. In this section we will 
construct an explicit example for an open system dynamics in a three-dimensional 
Hilbert space for which the optimal pair is not a pair of pure states.

We consider a $\Lambda$-system which interacts with an off-resonant cavity field. 
The weak-coupling master equation of this model is given by \cite{MeasurePaper}
\begin{eqnarray} \label{me}
\frac{d}{dt}\rho(t)&=&-i\lambda_1(t)[|a\rangle\langle a|,\rho(t)]-i\lambda_2(t)[|a
\rangle\langle a|,\rho(t)]\nonumber\\
&+&\gamma_1(t)\Big[|b\rangle\langle a|\rho(t)|a\rangle\langle b|-\frac{1}{2}\{\rho
(t),|a\rangle\langle a|\}\Big]\nonumber\\
&+&\gamma_2(t)\Big[|c\rangle\langle a|\rho(t)|a\rangle\langle c|-\frac{1}{2}\{\rho(t),|
a\rangle\langle a|\}\Big]\label{eq:6},
\end{eqnarray}
where $|a\rangle$ refers to the excited and $|b\rangle,|c\rangle$ to the two ground 
states. The coefficients $\lambda_{1,2}(t)$ and $\gamma_{1,2}(t)$ are determined
by the spectral density of the cavity field. Introducing the functions
\begin{eqnarray}
 f(t) &=& e^{-\bigl(D_1(t)+D_2(t)\bigr)/2}e^{-i\bigl(L_1(t)+L_2(t)\bigr)}, \\
 g_i(t) &=& \int_0^t ds \gamma_i(s)e^{-\bigl(D_1(s)+D_2(s)\bigr)},
\end{eqnarray}
where
\begin{equation}
 D_i(t) = \int_0^t ds \gamma_i(s), \qquad
 L_i(t) = \int_0^t ds \lambda_i(s),
\end{equation}
we find that the solution of the master equation yields the dynamical map
\begin{equation}
\Phi_t^\Lambda(\rho)= \begin{pmatrix}
				      |f(t)|^2\rho_{aa}	&	f(t)\rho_{ab}			&	f(t)
				      \rho_{ac}	\\
				      f(t)^*\rho_{ab}^*	&	g_1(t)\rho_{aa}+\rho_{bb}	&	
				      \rho_{bc}	\\
				      f(t)^*\rho_{ac}^*	&	\rho_{bc}^*			&	
				      g_2(t)\rho_{aa}+\rho_{cc}
			   \end{pmatrix}.
\end{equation}
The functions $f$, $g_1$ and $g_2$ have to obey the following relations 
which guarantee that $\Phi_t^\Lambda$ is trace preserving and completely 
positive,
\begin{eqnarray}
 && g_1(t)+g_2(t)+|f(t)|^2=1 \label{eq:CP},\\
 && g_{1,2}(t) \geq 0 \label{eq:T}.
\end{eqnarray}

We consider now, for simplicity, the case where the Lamb-shifts $\lambda_i$ 
are equal to zero, while the decay rates $\gamma_i$ in the dissipator of the 
master equation (\ref{me}) are chosen to obey one period of small oscillation,
\begin{equation}
 \lambda_{1,2}(t) = 0, \qquad \gamma_{1,2}(t) = 0.03 \cdot \sin(t).
\end{equation}
It is easy to check that for this choice the conditions \eqref{eq:CP} and 
\eqref{eq:T} are satisfied. We have carried out numerical simulations,
drawing random pairs of pure, orthogonal initial states and determining the 
corresponding increase of the trace distance for each initial pair. The results are 
shown in Fig.~\ref{fig:distrPure2}. We see from the figure that there is a finite
gap between the maximal possible increase of the trace distance for pure,
orthogonal initial pairs, and the increase of the trace distance corresponding
to the initial pair
\begin{equation} \label{mpair}
	\rho_1 = |a\rangle\langle a|, \qquad
	\rho_2 = \frac{1}{2}\big(|b\rangle\langle b| + |c\rangle\langle c|\big),
\end{equation} 
which consists of the excited state and the uniform mixture of the two ground 
states (indicated by the arrow in Fig.~\ref{fig:distrPure2}). Thus, this example 
clearly demonstrates that for Hilbert space dimensions larger than two, the optimal 
initial state pair can indeed contain a mixed state for certain non-Markovian 
dynamical maps \cite{note1}. 

\begin{figure}[tbh]
\centering  
\includegraphics[width=0.35\textwidth]{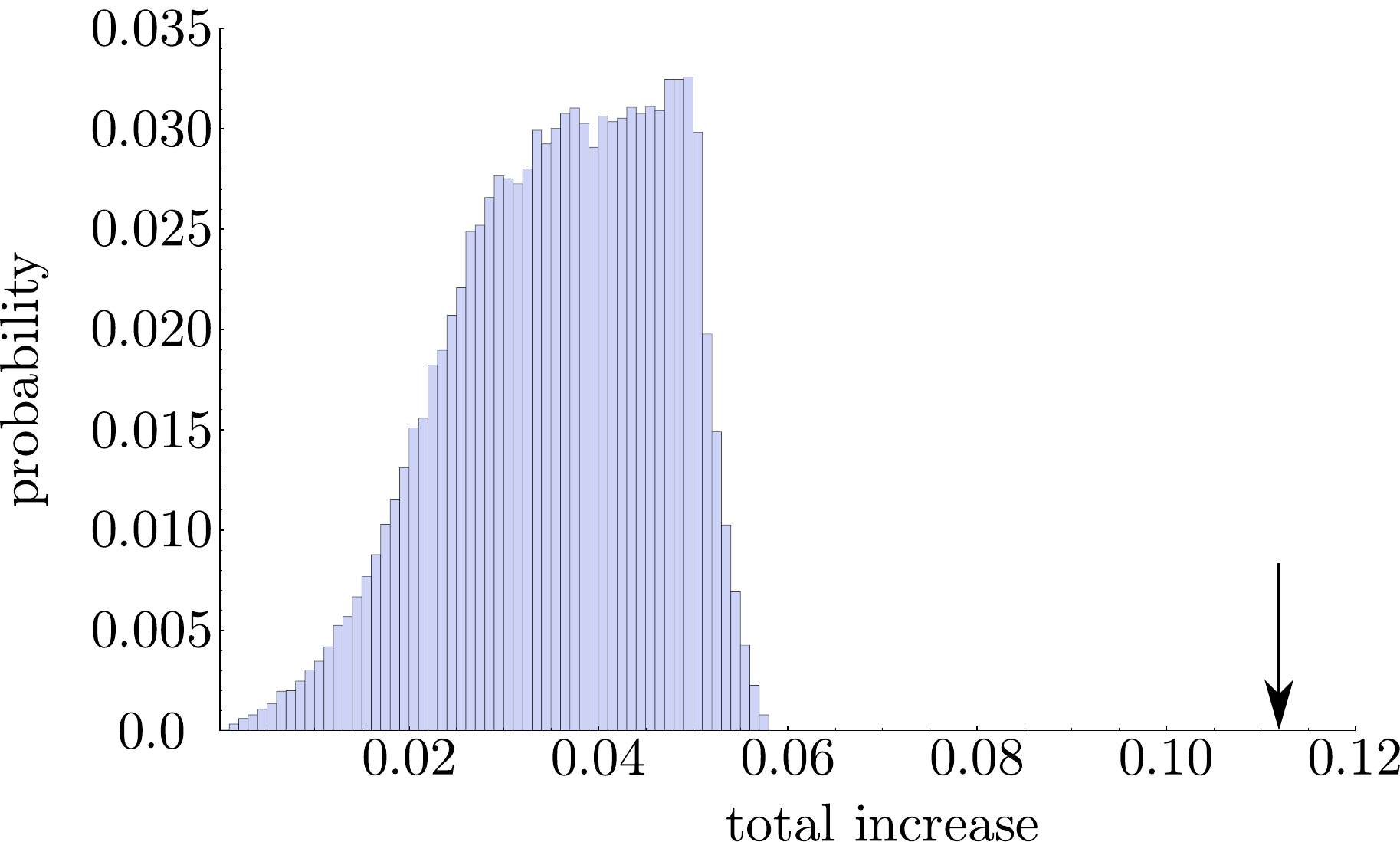}
\caption{(color online) Probability of the total increase of the trace distance for 
randomly drawn pairs of pure orthogonal states for one period of oscillation of the 
decay rates. The arrow indicates the value of the maximal increase of the trace 
distance for the pair of states given in Eq.~(\ref{mpair}) including a mixed state. 
The sample size of pairs of pure orthogonal states is equal to $10^5$.}
\label{fig:distrPure2}
\end{figure}

\section{Conclusion}\label{sec:conclusion}

In this paper we have shown that optimal pairs of initial states for 
non-Markovian quantum dynamics must be orthogonal which implies that these 
states can be distinguished with probability $1$ by a single measurement. Optimal 
state pairs thus have the maximal possible amount of initial information and are
therefore capable of emitting and reabsorbing the maximal amount of information
during the non-Markovian dynamics. 

We emphasize that the proof of these statements only relies on the convexity of 
the state space and on the linearity of the dynamical maps. Our results are thus
very general and can be applied to any quantum process describable by
a family of linear dynamical maps $\Phi_t$ in any Hilbert space. Strictly speaking, 
even the complete positivity of the dynamics is not needed. In fact, it suffices to 
assume that the maps $\Phi_t$ are positive since even trace-preserving positive 
maps are contractions for the trace distance \cite{Ruskai}. 

We have further demonstrated that for Hilbert spaces with dimensions of at least
$3$ optimal state pairs need not consist of pure states, in contrast to the case
of a qubit where optimal pairs are always antipodal points on the Bloch sphere, 
and as such pure. The example constructed here leads to an optimal pair 
consisting of a pure and a mixed quantum state. We conjecture that in Hilbert 
spaces of dimension $4$ and higher one can also construct quantum processes 
for which both states of the optimal pair are true mixtures. Finally, we mention that 
on the basis of the present results further, more specific statements could be 
proven if one assumes additional properties of the quantum process, such as 
invariance under certain symmetry groups, or the existence of invariant states. In 
addition, the notion of joint translatability provides new insights into the structure of 
the state space and might serve for further applications.

\acknowledgments
We gratefully acknowledge financial support by the German Academic Exchange 
Service (DAAD), the Magnus Ehrnrooth Foundation, the Jenny and Antti Wihuri 
Foundation, the Graduate School of Modern Optics and 
Photonics, and the Academy of Finland (mobility from Finland 259827). SW and
HPB would like to thank Bruno Leggio for helpful discussions.


\begin{thebibliography}{xx}
\bibitem{note1}
Note that incidentally  in Ref.~[9], the maximization was performed over the set of pure states only.
\end{thebibliography}
\end{document}